\documentclass{article}
\usepackage[utf8]{inputenc}

\usepackage{amsthm}
\usepackage{tikz}
\usepackage{listings}
\usepackage{courier}

\newtheorem{theorem}{Theorem}[section]
\newtheorem{proposition}[theorem]{Propostion}

\providecommand{\Omicron}{\mathrm{O}}

\title{Solving Shisen-Sho boards}
\author{Michiel de Bondt}

\begin{document}

\maketitle

\begin{abstract}
We give a simple proof of that determining solvability
of Shisen-Sho boards is NP-complete. Furthermore, we
show that under realistic assumptions, one can compute 
in logarithmic time if two tiles form a playable pair.

We combine an implementation of the algoritm to test
playability of pairs with my earlier algorithm to solve 
Mahjong Solitaire boards with peeking, to obtain an
algorithm to solve Shisen-Sho boards. We sample 
several Shisen-Sho and Mahjong Solitaire layouts
for solvability for Shisen-Sho and Mahjong Solitaire.
\end{abstract}

\section{Introduction}

Shisen-Sho and Mahjong Solitaire are games in which the object is to remove tiles in a certain way. The tiles come in groups of four equals. The game is
prepaired by stacking the tiles randomly in a certain layout, which may be threedimensional in the case of Mahjong Solitaire. The object is to remove the tiles in pairs of equal tiles. 

For Shisen-Sho and Mahjong Solitaire, the conditions under which a pair of equal tiles may be removed are different. For Shisen-Sho, a pair of equal tiles may be removed if they are either adjacent, or they can be connected by at most $3$ horizontal and vertical free lines. Here, a line is free if it does not cross or edge a tile. For Mahjong Solitaire, a pair of equal tiles may be removed if they are both free. A tile is free if there are no adjacent tiles above it, and neither on either the left or the right side.

Mahjong Solitaire is usually played with the 144 tiles of the Mahjong game,
where both the $4$ season tiles and the $4$ flower tiles are seen as a group
of $4$ equal tiles. Various layout shapes are used. Shisen-Sho is usually played with a rectangular layout.

A pair of tiles which can be removed is called a playable pair. One way to play the game is by randomly removing playable pairs, hoping that one will not get stuck. But if one gets stuck, one still does not know if all tiles could have be removed. To determine if all tiles could be removed or not, an exhaustive search is needed in general. An efficient way to do this for Mahjong Solitaire can be found in my paper \cite{dB}.  

For Shisen-Sho, we essentially use the same algorithm, but with the test for playable pairs replaced. There is however one point of attention. 
In the algorithm of \cite{dB}, there are scans used for pruning, in which the third and fourth tile of several groups may be removed individually as soon as they are free. 
For Mahjong Solitaire, this comes down to that any tile of such a group may be removed, as soon as it forms a playable pair with another tile of the group: a tile that may or may not have been removed already. But for Shisen-Sho, the latter is more restrictive, whence it leads to better pruning. So we use the last interpretation when we make a Shisen-Sho version of the algorithm.

But first, we prove that Shisen-Sho is NP-complete. More precisely, we prove
that determining if all tiles can be removed in a Shisen-Sho board with a rectangular layout is NP-complete. Next, we show that determining if
two equal tiles form a playable pair can be determined in logarithmic
time, under the assumption that we have bitwise operations on registers with at least as many bits as the layout dimensions.

After that, we extend the rules of the Shisen-Sho game to the Mahjong Solitaire grid, and adapt the algorithm to determine if two equal tiles form a playable pair to this grid. At the end, we present some graphs
to indicate how many random boards are winnable (i.e.\@ all tiles can be
removed) for several layouts, for Shisen-Sho, Mahjong Solitaire, and transposed Mahjong Solitaire. With transposed Mahjong Solitaire, the layout
is transposed, or equivalently, left and right are replaced by front
and rear in the rules of Mahjong Solitaire.

\section{Complexity of Shisen-Sho}

In \cite{IWM}, the following theorem has been proved already, but the
proof is longer than the proof below. The refinement of having only
$5$ rows seems a new result.

\begin{theorem}
Shisen-Sho is NP-complete. 
\end{theorem}

\begin{proof}
We reduce from Mahjong Solitaire with peeking with isolated stacks.
The NP-completeness of that is proved under the tag Shanghai in \cite{Ep}.

Let $a$ be a tile of a Mahjong Solitaire stack. We turn $a$ in a supertile
of size $5 \times 5$ for Shisen-Sho by adding $6$ extra groups.
\begin{center}
\begin{tikzpicture}[x=5mm,y=5mm]
\fill[black!10] (0,0) rectangle (5,5);
\fill[black!25] (0,0) rectangle (2,4) (3,0) rectangle (5,4);
\fill[white] (2,2) rectangle (3,3);
\draw[black!50] \foreach \x in {1,...,4} { (\x,0) -- (\x,5) };
\draw[black!50] \foreach \y in {1,...,4} { (0,\y) -- (5,\y) };
\draw (0,0) rectangle (5,5) (0,0) rectangle (2,2) (3,0) rectangle (5,2);
\draw (0,2) rectangle (2,4) (2,2) rectangle (3,3) (3,2) rectangle (5,4);
\draw \foreach[count=\n] \v in {z,z,z,z',z'} { (\n-1,4.5) node[anchor=west] {$\strut\v$} };
\draw \foreach[count=\n] \v in {x,x',z',y',y} { (\n-1,3.5) node[anchor=west] {$\strut\v$} };
\draw \foreach[count=\n] \v in {x',x,a,y,y'} { (\n-1,2.5) node[anchor=west] {$\strut\v$} };
\draw \foreach[count=\n] \v in {y',y,z,x,x'} { (\n-1,1.5) node[anchor=west] {$\strut\v$} };
\draw \foreach[count=\n] \v in {y,y',z',x',x} { (\n-1,0.5) node[anchor=west] {$\strut\v$} };
\end{tikzpicture}
\qquad
\begin{tikzpicture}[x=5mm,y=5mm]
\fill[black!10] (0,0) rectangle (5,5);
\fill[black!25] (0,0) rectangle (2,4) (3,0) rectangle (5,4);
\draw[black!50] \foreach \x in {1,...,4} { (\x,0) -- (\x,5) };
\draw[black!50] \foreach \y in {1,...,4} { (0,\y) -- (5,\y) };
\fill[white] (2,2) rectangle (3,5) (1,4) rectangle (4,5);
\draw (0,4) rectangle (1,5) (4,4) rectangle (5,5); 
\draw (2,0) -- (3,0) (0,0) rectangle (2,2) (3,0) rectangle (5,2);
\draw (0,2) rectangle (2,4) (2,2) rectangle (3,3) (3,2) rectangle (5,4);
\draw \foreach[count=\n] \v in {z,,,,z'} { (\n-1,4.5) node[anchor=west] {$\strut\v$} };
\draw \foreach[count=\n] \v in {x,x',,y',y} { (\n-1,3.5) node[anchor=west] {$\strut\v$} };
\draw \foreach[count=\n] \v in {x',x,a,y,y'} { (\n-1,2.5) node[anchor=west] {$\strut\v$} };
\draw \foreach[count=\n] \v in {y',y,z,x,x'} { (\n-1,1.5) node[anchor=west] {$\strut\v$} };
\draw \foreach[count=\n] \v in {y,y',z',x',x} { (\n-1,0.5) node[anchor=west] {$\strut\v$} };
\end{tikzpicture}
\end{center}
In order to remove one of the $x$, $x'$, $y$ and $y'$-tiles, the $a$-tile needs be removed. 
Hence the $a$-tile needs to be removed to free the $z$-tile and $z'$-tile below it as well.
So the $z$-tile and $z'$ above the $a$-tile need to be removed to free the $a$-tile. This can 
be done (in more ways if there is no supertile on top of the supertile of $a$), 
after which the $a$-tile has a connection to above. But the connection is blocked
if there is another supertile on top of the supertile of $a$. If tile $a$ can be removed 
somehow, then the whole supertile of $a$ can be freed.

So by stacking supertiles, we can emulate Mahjong Solitaire with peeking with isolated 
stacks in Shisen-Sho. The stacks of supertiles can be filled up to a rectangle, by adding
tile groups in literal order. This is because those tiles can be played in the same
order in accordance with the rules of Shisen-Sho.
\end{proof}

\begin{theorem}
Shisen-Sho with $5$ rows is NP-complete. 
\end{theorem}

\begin{proof}
We reduce from Mahjong Solitaire with peeking with isolated stacks of the form
$aab$ and $abb$. The NP-completeness of that is proved in \S 2 of \cite{dB}.

We turn stacks $aab$ and $abb$ as follows into tileblocks.
\begin{center}
\begin{tikzpicture}[x=5mm,y=5mm]
\fill[black!10] (0,0) rectangle (5,5);
\fill[black!25] (0,0) rectangle (2,4) (3,0) rectangle (5,4);
\draw[black!50] \foreach \x in {1,...,4} { (\x,0) -- (\x,5) };
\draw[black!50] \foreach \y in {1,...,4} { (0,\y) -- (5,\y) };
\draw[fill=white] (2,1) rectangle (3,2) (2,2) rectangle (3,3) (2,4) rectangle (3,5);
\draw (0,0) rectangle (5,5) (0,0) rectangle (2,2) (3,0) rectangle (5,2);
\draw (0,2) rectangle (2,4) (2,2) rectangle (3,3) (3,2) rectangle (5,4);
\draw \foreach[count=\n] \v in {z,z',a,z',z} { (\n-1,4.5) node[anchor=west] {$\strut\v$} };
\draw \foreach[count=\n] \v in {x,x',z,y',y} { (\n-1,3.5) node[anchor=west] {$\strut\v$} };
\draw \foreach[count=\n] \v in {x',x,a,y,y'} { (\n-1,2.5) node[anchor=west] {$\strut\v$} };
\draw \foreach[count=\n] \v in {y',y,b,x,x'} { (\n-1,1.5) node[anchor=west] {$\strut\v$} };
\draw \foreach[count=\n] \v in {y,y',z,x',x} { (\n-1,0.5) node[anchor=west] {$\strut\v$} };
\end{tikzpicture}
\qquad
\begin{tikzpicture}[x=5mm,y=5mm]
\fill[black!10] (0,0) rectangle (5,5);
\fill[black!25] (0,0) rectangle (2,4) (3,0) rectangle (5,4);
\draw[black!50] \foreach \x in {1,...,4} { (\x,0) -- (\x,5) };
\draw[black!50] \foreach \y in {1,...,4} { (0,\y) -- (5,\y) };
\draw[fill=white] (2,1) rectangle (3,2) (2,3) rectangle (3,4) (2,4) rectangle (3,5);
\draw (0,0) rectangle (5,5) (0,0) rectangle (2,2) (3,0) rectangle (5,2);
\draw (0,2) rectangle (2,4) (2,2) rectangle (3,3) (3,2) rectangle (5,4);
\draw \foreach[count=\n] \v in {z,z',a,z',z} { (\n-1,4.5) node[anchor=west] {$\strut\v$} };
\draw \foreach[count=\n] \v in {x,x',b,y',y} { (\n-1,3.5) node[anchor=west] {$\strut\v$} };
\draw \foreach[count=\n] \v in {x',x,z,y,y'} { (\n-1,2.5) node[anchor=west] {$\strut\v$} };
\draw \foreach[count=\n] \v in {y',y,b,x,x'} { (\n-1,1.5) node[anchor=west] {$\strut\v$} };
\draw \foreach[count=\n] \v in {y,y',z,x',x} { (\n-1,0.5) node[anchor=west] {$\strut\v$} };
\end{tikzpicture}
\end{center}
The tiles with $z'$ are for filling up, and are shared among two tileblocks.
\end{proof}

\section{Testing if two tiles are a matching pair}

Suppose the width and height of the board are estimated by $n$. 
Let $d$ be the distance between the two tiles at hand. 

\begin{proposition}
We can test if two tiles are a matching pair in $\Omicron(n \cdot d)$ steps. 
\end{proposition}

\begin{proof}
We only need to study connections with the first and third line horizontal and 
the second line vertical, since the other type of three-line connection is similar,
and connections with fewer lines are degenerate cases of three-line connections.

First, we compute how far we can reach horizontally from both tiles. This can be done
in $\Omicron(n)$ steps. For each column that can be reached horizontally by both tiles, 
we check if it is free between the two tiles. This can be done in $\Omicron(d)$ steps 
for each column. 
\end{proof}

We can improve the bound of $\Omicron(n \cdot d)$ drastically, if we assume that we have 
registers with $n$ bits. 
We assume that we can perform bitwise logic on registers with $n$ bits, including zero testing. 
Furthermore, we assume that we can set the first $k$ bits and clear the last $n-k$ bits 
of a register with $n$ bits. The value of $k$ comes from a smaller register which can hold the
values $0,1,2,\ldots,n$, i.e.\@ a register with more than $\log n$ bits. We assume
that we can do everything on the small registers which we can do on the large registers
with $n$ bits. Furthermore, we assume that we can do either addition or subtraction 
(hence both) on the small registers. Our memory usage is $\Omicron(n)$ registers.

\begin{theorem} \label{notree}
Testing if two tiles form a matching pair can be performed in $\Omicron(\log n + d)$ steps. 
With bit scan techniques available on the large registers, only $\Omicron(d)$ steps are
needed.
\end{theorem}

\begin{proof}
Again, we only need to study connections with the first and third line horizontal and 
the second line vertical. We assume that we have registers with tile population vectors 
for each row.

Computing how far one can reach horizontally from a tile can be done with binary
search in $\Omicron(\log n)$ steps. With bit scan techniques, $\Omicron(1)$ steps 
are sufficient.
Finding a free column in the common horizontal range can be done in parallel, by
computing the bitwise disjunction (bitwise or) of the rows between both tiles.
This takes $\Omicron(d)$ steps.
\end{proof}

In a practical algorithm where 64 bit integers are used as large registers, 
computing how far one can reach horizontally from a tile will be done differently,
because subtraction can be used as well. In our algorithm, we first remove the tile 
at hand from the population vector. Next, we compute how far one can reach horizontally
from the position of the left side of the tile, using the following function.

\begin{lstlisting}[language=C++,basicstyle=\ttfamily\footnotesize,
backgroundcolor=\color{black!10},frame=shadowbox,
rulecolor=\color{black!50},rulesepcolor=\color{black!25},xrightmargin=2pt]
// optional speedup
#define bitscanreverse64(g) (63 ^ __builtin_clzl(g))

// p is the position to scan from for consecutive cleared bits in f
// p must be in the interval [0..63], optional to allow 64 as well
// consecutive cleared bits found are returned as set bits
static unsigned long fillrange (unsigned long f, int p)
{
  unsigned long g = f;
  // /* if (0 <= p <= 63) */ g = bits of f below position p
  /* if (!(p & ~63)) */ g &= ((1lu << p) - 1);
  if (g == 0) return ~f & (f - 1);
  // g = smallest power of 2 exceeding g
#ifdef bitscanreverse64
  g = 2lu << bitscanreverse64 (g);
#else
  g |= g >> 1;
  g |= g >> 2;
  g |= g >> 4;
  g |= g >> 8;
  g |= g >> 16;
  g |= g >> 32;
  g++;
#endif
  return ~f & (f - g);
}
\end{lstlisting}

Notice that forward bit scanning is not necessary, because subtraction can be used. 
If reverse bit scanning is not available, then binary search is still not used, 
because it takes conditional jumps, which is not preferable on modern CPUs.
Just as binary search, the actual replacement technique for reverse bit scanning takes 
$\Omicron(\log n)$ steps. 

If one replaces {\ttfamily\^{ }} by {\ttfamily-} in the macro {\ttfamily bitscanreverse64}, 
then gcc will not optimize it to the corresponding Intel/AMD x64 instruction.

\begin{theorem}
If we allow playing and unplaying a tile to take $\Omicron(\log n)$ steps as a gambit, then testing 
if two tiles form a matching pair can be done in $\Omicron(\log n)$ steps. 
With bit scan techniques available on all registers, this improves to $\Omicron(\log d)$ 
steps.
\end{theorem}

\begin{proof}
We improve the computation of the bitwise disjunction (bitwise or) in the proof of 
theorem \ref{notree} from $\Omicron(d)$ to $\Omicron(\log d)$. 

Suppose that the rows are indexed $0,1,\ldots,n-1$. We embed the row indexes in
the odd numbers of the range $1,2,\ldots,2n-1$, by sending row index $k$ to $2k+1$.
We see the rows as leaves of a perfect tree. The even numbers in the range 
$1,2,\ldots,2n-1$ are used for the interior nodes of the tree, as depicted below.
\begin{center}
\begin{tikzpicture}[x=10.5pt,y=10.5pt]
\tikzstyle{nodestyle}=[circle,draw,fill=black!10,inner sep=0pt,minimum size=4mm]
\draw (32,5) node[nodestyle] (32) {$\scriptstyle 32$};
\draw (16,4) node[nodestyle] (16) {$\scriptstyle 16$} (16) -- (32);
\foreach \n/\m in {8/16,24/16} {
  \draw (\n,3) node[nodestyle] (\n) {$\scriptstyle \n$} (\n) -- (\m);
}
\foreach \n/\m in {4/8,12/8,20/24,28/24} {
  \draw (\n,2) node[nodestyle] (\n) {$\scriptstyle \n$} (\n) -- (\m);
}
\foreach \n/\m in {2/4,6/4,10/12,14/12,18/20,22/20,26/28,30/28} {
  \draw (\n,1) node[nodestyle] (\n) {$\scriptstyle \n$} (\n) -- (\m);
}
\foreach \n/\m in {1/2,3/2,5/6,7/6,9/10,11/10,13/14,15/14,
                   17/18,19/18,21/22,23/22,25/26,27/26,29/30,31/30} {
  \draw (\n,0) node[nodestyle,fill=black!25] (\n) {$\scriptstyle \n$} (\n) -- (\m);
}
\end{tikzpicture}
\end{center}
The data of the leaves are the population vectors of the corresponding rows.
The data of the internal nodes are the bitwise disjunction of the data of
the child nodes, i.e.\@ the bitwise disjunction of all population vectors
of its descendant leaves. This can indeed be maintained in $\Omicron(\log n)$
steps when a tile is played or unplayed, see also the code below.

Suppose that we have tiles in rows $i-1$ and $j$ such that $i-1 < j$. Then
we must compute the bitwise disjunction of rows $i,i+1,\ldots,j-1$,
i.e. the bitwise disjunction of the leaves between $2i$ and $2j$.
If $i < j$, then the index of the last common ancestor of leaves $2i + 1$ and $2j + 1$ 
lies between $2i + 1$ and $2j + 1$. By taking $k$ such that $2k$ is the index 
of this common ancestor if $i < j$, and by taking $k = i$ if $i = j$, we obtain 
that $2k$ lies between $2i - 1$ and $2j + 1$, and is the index of a common 
ancestor of leaves $2i + 1$ and $2j - 1$ if $i < j$, which we assume from now on.

The computation of $k$ such that $2k$ is the last common ancestor of leaves $2i+1$ and $2j+1$, where $i < j$, can be done in $\Omicron(1)$ steps with reverse bit scanning on small registers, see the code below. Without bit scan techniques, we can use binary search to replace reverse bit scanning, just as for the large registers in the computation of the horizontal range, and $\Omicron(\log \log n)$ steps suffice.

We can compute the bitwise disjunction of the leaves between $2i$ and $2k$ in $\Omicron(\log d)$ steps as follows. Suppose first that both $i$ and $k$ are odd. Then node $2k$ is the parent node of node $2i+1$, and 
either $k = i$ or $k = i+1$. So $k = i$ and the result is $0$.

Suppose next that either $i$ or $k$ is odd, but not both. Then node $2i + 2$ the last ancestor of node $2i+3 = 2(i+1) + 1$ which is smaller than $2i+3$, and $i + 1 \le k$. Furthermore, node $2i + 2$ is the last ancestor of node $2i+1$ which is larger than $2i+1$, and node $2k$ is any such ancestor, so node $2k$ is an ancestor of or just equal to node $2i + 2$. Consequently,
node $2k$ is an ancestor of node $2(i+1) + 1$. We can obtain the
bitwise disjunction of the leaves between $2(i+1)$ and $2k$ recursively.
The result is the bitwise disjunction of that and the data of node $2i+1$.

Suppose finally that both $i$ and $k$ are even. Then node $2i + 2 = 2\big(2(i/2)+1\big)$ is the parent of node $2i + 1$. So node $2k = 2\big(2(k/2)\big)$ is an ancestor of node $2\big(2(i/2)+1\big)$. First, we remove the leaves of the tree and keep the internal nodes with their even numbers. Next, we divide the node numbers by $2$, to obtain a new tree of the same type as the one we started with. The result is the bitwise disjunction of the new leaves between $2(i/2)$ and $2(k/2)$ of the new tree, which we can obtain recursively.

The bitwise disjunction of the leaves between $2k$ and $2j$ can be computed similarly.
\end{proof}

Below is the code of a practical algorithm where 64 bit integers are used as 
large registers, for updating the bitwise disjunction tree after playing or 
unplaying a tile.

\begin{lstlisting}[language=C++,basicstyle=\ttfamily\footnotesize,
backgroundcolor=\color{black!10},frame=shadowbox,
rulecolor=\color{black!50},rulesepcolor=\color{black!25},xrightmargin=2pt]
// updates fill[2], fill[4], ..., fill[126] as
// interlacing ortree of fill[1], fill[3], ..., fill[127]
// after changing fill[p] for odd p between 0 and 128
//                                 ______________________________32
//                 ______________16______________
//         ______08______                  ______24______
//     __04__          __12__          __20__          __28__
//   02      06      10      14      18      22      26      30
// 01  03  05  07  09  11  13  15  17  19  21  23  25  27  29  31
static void ortree_update (unsigned long *fill, int p)
{
  p &= -4;
  fill[p+2] = fill[p+1] | fill[p+3];
  p &= -8;
  fill[p+4] = fill[p+2] | fill[p+6];
  p &= -16;
  fill[p+8] = fill[p+4] | fill[p+12];
  p &= -32;
  fill[p+16] = fill[p+8] | fill[p+24];
  // the below can be skipped because fill[1] and fill[127]
  // are never taken as or-operand in applications of ortree_or
  /* p &= 64;
  fill[p+32] = fill[p+16] | fill[p+48];
  p = 0;
  fill[p+64] = fill[p+32] | fill[p+96]; */
}
\end{lstlisting}

Below is the code of a practical algorithm where 64 bit integers are used as 
large registers, for the actual computation of the bitwise disjunctions.

\begin{lstlisting}[language=C++,basicstyle=\ttfamily\footnotesize,
backgroundcolor=\color{black!10},frame=shadowbox,
rulecolor=\color{black!50},rulesepcolor=\color{black!25},xrightmargin=2pt]
// returns fill[2*p1+1] | fill[2*p1+3] | ... | fill[2*p2-1]
// assuming that fill[2], fill[4], ..., fill[126] is
// interlacing ortree of fill[1], fill[3], ..., fill[127]
//                                 ______________________________32
//                 ______________16______________
//         ______08______                  ______24______
//     __04__          __12__          __20__          __28__
//   02      06      10      14      18      22      26      30
// 01  03  05  07  09  11  13  15  17  19  21  23  25  27  29  31
static unsigned long ortree_or_x2 (unsigned long *fill, int p1, int p2)
{
  // negation of leading bit except value -1 for 0
  static int negationofleadingbit[128] = {
     -1, -1, -2, -2, -4, -4, -4, -4, -8, -8, -8, -8, -8, -8, -8, -8,
    -16,-16,-16,-16,-16,-16,-16,-16,-16,-16,-16,-16,-16,-16,-16,-16,
    -32,-32,-32,-32,-32,-32,-32,-32,-32,-32,-32,-32,-32,-32,-32,-32,
    -32,-32,-32,-32,-32,-32,-32,-32,-32,-32,-32,-32,-32,-32,-32,-32,
    -64,-64,-64,-64,-64,-64,-64,-64,-64,-64,-64,-64,-64,-64,-64,-64,
    -64,-64,-64,-64,-64,-64,-64,-64,-64,-64,-64,-64,-64,-64,-64,-64,
    -64,-64,-64,-64,-64,-64,-64,-64,-64,-64,-64,-64,-64,-64,-64,-64,
    -64,-64,-64,-64,-64,-64,-64,-64,-64,-64,-64,-64,-64,-64,-64,-64
  };
  // if p1 == p2, then p1 == p == p2
  // if p1 < p2, then p1 + 1 <= p <= p2, and if we see tree as family
  // tree, then p is the last common ancestor of p1 + 1 and p2, and
  // 2 * p is the last common ancestor of 2 * p1 + 1 and 2 * p2 + 1
  int p = p2 & negationofleadingbit[ p1 ^ p2 ];
  unsigned long f = 0;
  unsigned long *fill_2p = fill + 2 * p;

  // f |= fill[2*p1+1] | fill[2*p1+3] | ... | fill[2*p-1]
  p1 -= p;
  // f |= fill_2p[2*p1+1] | fill_2p[2*p1+3] | ... | fill_2p[-1]
  while (p1) {
    // s = trailing bit of p1
    int s = p1 & -p1;
    f |= fill_2p[ 2 * p1 + s ];
    p1 += s;
  }

  // f |= fill[2*p+1] | fill[2*p+3] | ... | fill[2*p2-1]
  p2 -= p;
  // f |= fill_2p[1] | fill_2p[3] | ... | fill_2p[2*p2-1]
  while (p2) {
    // s = trailing bit of p2
    int s = p2 & -p2;
    p2 -= s;
    f |= fill_2p[ 2 * p2 + s ];
  }

  return f;
}
\end{lstlisting}

In the computation 
of $k$, the usage of bit scan reverse is replaced by a lookup table. Our theoretical
model allows such a table as well, because it consists of $\Omicron(n)$ small 
registers. The initialization of the table can be done in $\Omicron(n)$ steps.

There are a few low level optimizations in the code. The most important of them
is that the array pointer is translated to make computations easier. 
The loops are controlled by zero testing on variables, so no extra
comparisons are required. One might think that the array index
computations in the loops take $3$ 
instructions on Intel/AMD x64 machines, namely either a move, a shift (or 
multiplication) and an addition, or a move and $2$ additions. But only one 
instruction is needed, namely a `load effective address' instruction.

\section{The Mahjong grid}

Mahjong Solitaire has a more complex tile positioning system. First, tiles have three
coordinates, namely row, column, and level. Second, the row and column coordinates
are a multiple of $\frac12$ instead of $1$. The latter can also be seen as a property 
of the tiles themselves, namely that their dimensions are $2 \times 2 \times 1$.

\subsection{Shisen-Sho in three dimensions}

We extend the Shisen-Sho game to three dimensions as follows. If there are other tiles above a given tile, then that tile cannot be played. So the pillar above the tile to be played must be free. If we have two tiles on the same level with no tiles above them, then they can be paired if they can paired within their level.
If we have two tiles on the same level with no tiles above them, then they can be 
paired if they connect with at most $2$ horizontal and vertical lines in the highest
level. 

The pillar line above the lowest tile can be seen as the third line. If we would allow $2$ 
of the $3$ lines to be pillar lines, then tiles on the same level may be pairable in the 
threedimensional grid without being pairable within their level. The result of that
would be that the threedimensional game is not an extension of the twodimensional game,
which is not what we want. 

So we only allow only $1$ of the $3$ lines to be a pillar line. If tile population
is descending in every pillar, then tiles on different levels can be connected by three 
grid lines of which one pillar line, if and only if this can be done in the above way, 
i.e.\@ by starting with a pillar line from the lowest tile. The
xmahjongg layouts indeed have the property that tile population
is descending in every pillar during game play.

\subsection{Pair matching tests for the Mahjong grid}

Besides tile population vectors for the rows and columns, tile population vectors for
the pillars are used. We use 16 bit integers for these vectors, so that 16 levels can 
be distinguished. The tile population vectors for the pillars are in agreement with
the real situation, so every tile contributes to $4$ tile population vectors for the 
pillars. 

But the tile population vectors for the rows and columns are not in agreement with
the real situation any more. If a tile is unplayed, then $2$ consecutive bits are
set in only $1$ row instead of $2$, and only $1$ column instead of $2$. So the 
real tile population vector for row $i$ is the bitwise disjunction of the used 
tile population vectors for rows $i$ and $i + 1$. The same holds for the columns.

The horizontal range of a tile is the bitwise disjunction of the horizontal range of 
the tile in its $2$ rows, as real tile population vectors. The actual computation
of it in our code is as follows.
\begin{lstlisting}[language=C++,basicstyle=\ttfamily\footnotesize,
backgroundcolor=\color{black!10},frame=shadowbox,
rulecolor=\color{black!50},rulesepcolor=\color{black!25},xrightmargin=2pt]
      // row above tile
      unsigned long f1 = rowfill[lev][r-2];
      // row of tile without tile itself
      unsigned long f2 = rowfill[lev][r] & ~(3lu << col);
      // row below tile
      unsigned long f3 = rowfill[lev][r+2];
      f1 |= f2;
      f3 |= f2;
      if (f1 == f2 || f2 == f3) {
        rowfillrange = fillrange(f2,col);
      } else {
        rowfillrange = fillrange(f1,col) | fillrange(f3,col);
      }
\end{lstlisting}

Since the tile appears in only one tile population vector for the rows, 
it only needs to be removed from one such vector. The result of this
is {\ttfamily f2}. {\ttfamily f1} and {\ttfamily f3} are tile 
population vectors for real rows, again with the tile at hand removed. 
The indexes {\ttfamily r-2} and {\ttfamily r+2} (instead of 
{\ttfamily r-1} and {\ttfamily r+1}) are because of the interlacing 
internal nodes of the bitwise disjunction tree.

There is some testing to ensure that is some cases where this is possible, 
only one horizontal range computation is performed. One such case is where 
the Mahjong grid is used to emulate the Shisen-Sho grid by way of using 
only even row and column positions. 

Surprisingly, we can find connections with the first and third line horizontal 
and the second line vertical in the same way as for the Shisen-Sho grid,
i.e. we test if the bitwise conjunction of the horizontal ranges of the tiles
(as sequences of set bits) conincides with a free bit of the bitwise disjunction 
of the tile population vectors of the rows in between the tiles. There are $2$ 
differences with respect to the Shisen-Sho grid here: the tile population vectors
of the rows are not in agreement with the real situation, and the computation 
of the horizontal ranges is different.

Other types of connections between tiles on the same level are found in a similar
way. For connections between tiles on a different level, the horizontal and vertical 
range of the lowest tile is replaced by just the tile itself, after which the 
connectivity of the tiles is tested in the level of the highest tile just as for 
tiles on the same level.

There are $2$ versions of the program: one with the Shisen-Sho grid and one
with the Mahjong grid. But in contrast to the above explanation, the version
with the Shisen-Sho grid adopts three dimensions already. The version for
the Mahjong grid is slower, but not very much, which is in agreement with the $\Omicron(1)$ extra complexity of the pair matching computations.

\section{Sampling results}

We sampled several layouts of xmahjongg and several rectangular layouts. 
To be able to sample Shisen-Sho on xmahjongg layouts, we had to use the 
Shisen-Sho solver for the Mahjong grid. We used this solver for the rectangular grids as well, so the width and height was limited by $32$.

We made a new Mahjong Solitaire solver along with the Shisen-Sho solver (as a extra compiler option). This solver differs from the original Mahjong Solitaire solver in that it requires free tiles not to have any tile above them, instead of just adjacent tiles. But for layouts for which tile polulation is decreasing in every pillar, such as the xmahjongg layouts, this difference has no effect on the game. 

Since the new Mahjong Solitaire solver is about $3$ times as slow as the original Mahjong Solitaire solver, we used the original Mahjong Solitaire solver for the Mahjong Solitaire sampling results. 

For Shisen-Sho, we chose $1.15^{6 \sqrt{k}}$ as the number of attempts to solve a board by playing random matches, with $k$ being the number of 
groups with $4$ equal tiles. For Mahjong Solitaire, we chose
$1.2^{6 \sqrt{k}}$ to be this number. This number was
$1.2^k$ originally, which comes down to the same for layouts with $144$
tiles, such as the xmahjongg layouts.

The foo layout and the bar layout are not part of xmahjongg. They were
designed by me, and are as follows.
\begin{center}
\begin{tikzpicture}[x=2.5mm,y=2.5mm]
\draw[fill=lightgray] (0,0) rectangle (18,20);
\draw \foreach \x in {2,4,...,16} {
   (\x,0) -- (\x,20)
};
\draw \foreach \y in {2,4,...,18} {
   (0,\y) -- (18,\y)
};
\draw[fill=lightgray] (3,3) rectangle (15,17);
\draw \foreach \x in {5,7,...,13} {
   (\x,3) -- (\x,17)
};
\draw \foreach \y in {5,7,...,15} {
   (3,\y) -- (15,\y)
};
\draw[fill=lightgray] (6,6) rectangle (12,14);
\draw \foreach \x in {8,10} {
   (\x,6) -- (\x,14)
};
\draw \foreach \y in {8,10,12} {
   (6,\y) -- (12,\y)
};
\begin{scope}[shift={(23,1)}]
\draw[fill=lightgray] (0,0) rectangle (20,18);
\draw \foreach \x in {2,4,...,18} {
   (\x,0) -- (\x,18)
};
\draw \foreach \y in {2,4,...,16} {
   (0,\y) -- (20,\y)
};
\draw[fill=lightgray] (3,3) rectangle (17,15);
\draw \foreach \x in {5,7,...,15} {
   (\x,3) -- (\x,15)
};
\draw \foreach \y in {5,7,...,13} {
   (3,\y) -- (17,\y)
};
\draw[fill=lightgray] (6,6) rectangle (14,12);
\draw \foreach \x in {8,10,12} {
   (\x,6) -- (\x,12)
};
\draw \foreach \y in {8,10} {
   (6,\y) -- (14,\y)
};
\end{scope}
\end{tikzpicture}
\end{center}
The foo and the bar layout have $3\cdot 4  + 6\cdot 7 + 9\cdot 10 = 144$
tiles, just like the xmahjongg layouts.
The foo and the bar layout are the transpose of each other. The deepwell layout is selftranspose/symmetric.

We used the Computational Science compute cluster of Radboud University Nijmegen for our layout sampling. The computations took about $5$ weeks. The longest computations run for about $3$ weeks. All Shisen-Sho computations were done in one run, but many of the (transposed) Mahjong Solitaire computations were aborted by system administration. For that reason, I had to write versions of the program for (transposed) Mahjong Solitaire which resume aborted computations. Computation time could not be restored, since it is only written at the end of the computation.

We will present graphs of the samples: the exact sample results can be found in the source along with the program code.

\begin{figure*}[!p]
{\center \large \bf 
Proportion of winnable Shisen-Sho games for \\ 
rectangular layouts and xmahjongg layouts \\
}
\bigskip
\makebox[\textwidth][c]{\includegraphics{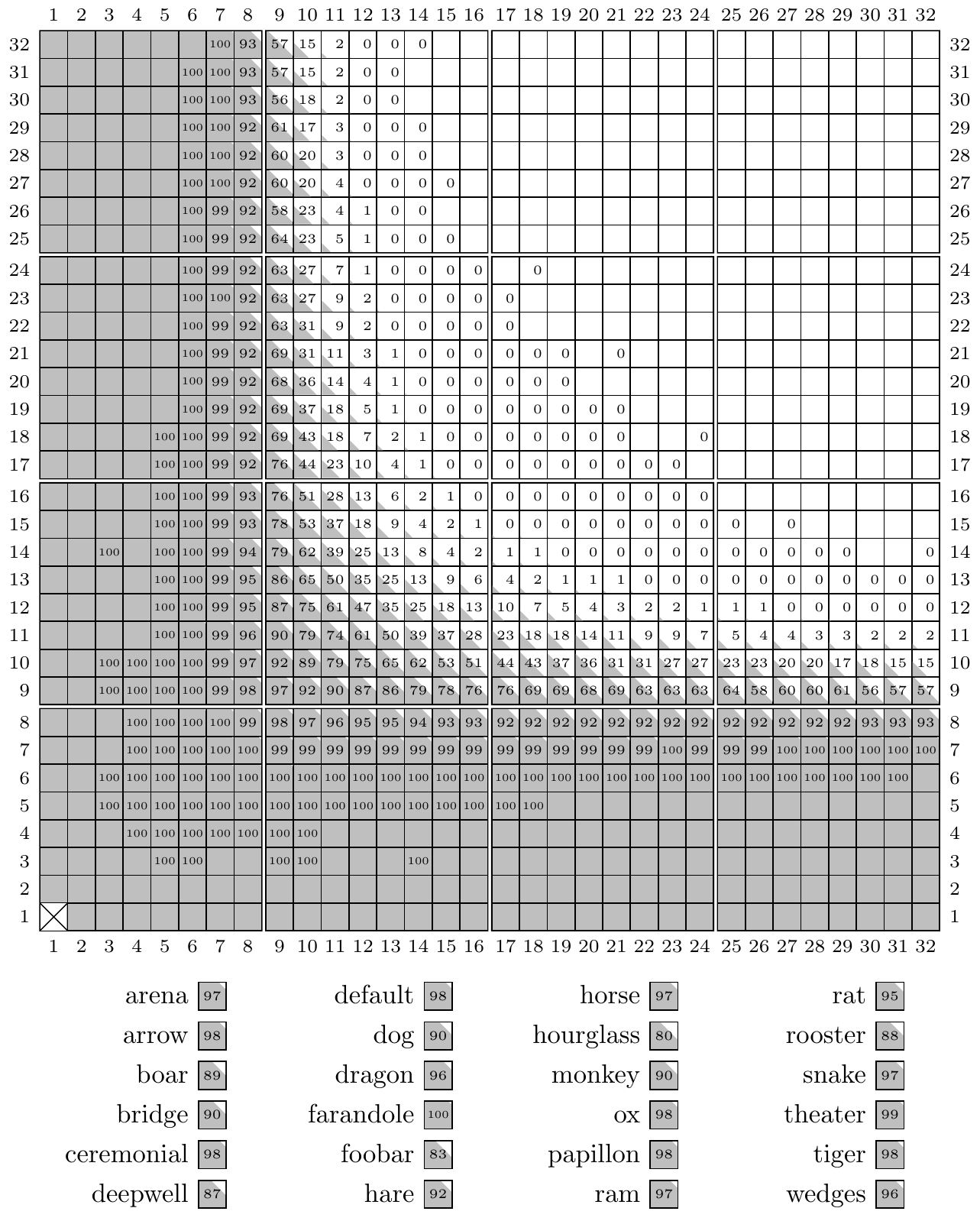}}
\end{figure*}

\begin{figure*}[!p]
{\center \large \bf 
Proportion of winnable Mahjong Solitaire games for \\ 
rectangular layouts and xmahjongg layouts \\
}
\bigskip
\makebox[\textwidth][c]{\includegraphics{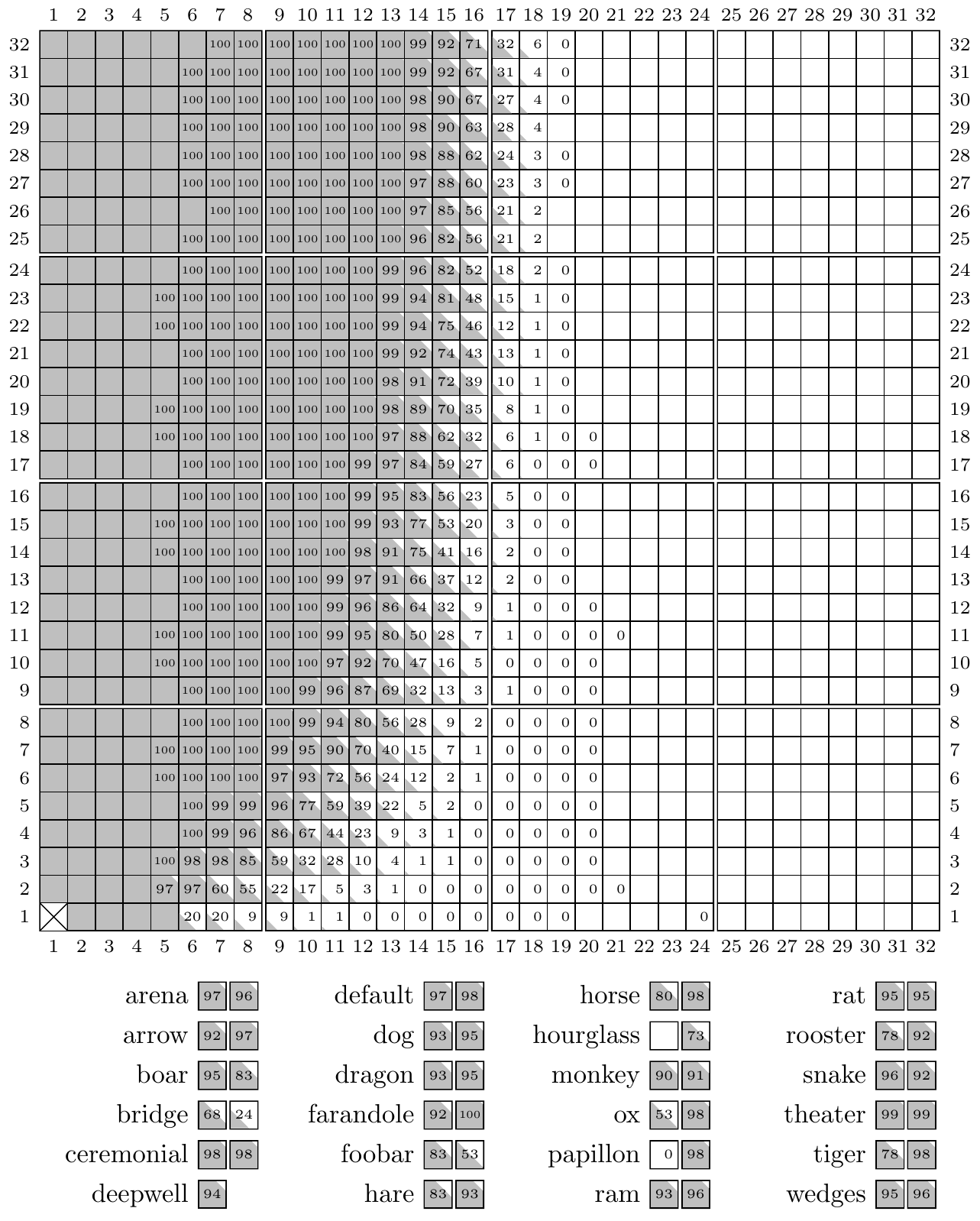}}
\end{figure*}

\subsection{Rectangular layouts}

We sampled rectangular layouts up to $32 \times 32$ inclusive. For
layouts with an odd number of tiles, we removed one corner tile. For the $1 \times 1$ layout, this resulted in a void layout, which we did not sample.
For other layouts with $2$ or $3$ tiles modulo $4$, we used one tile group
with only $2$ equal tiles.

For Shisen-Sho, the sample sizes are given on the left. For Mahjong Solitaire, these sample sizes were too time-consuming, so we adapted them as indicated on the right.
\begin{center}
\begin{tikzpicture}[x=0.5in,y=0.5in]
\draw[fill=black!40] (0.125,0.125) -- (0,0.125) -- 
    (0,1) -- (1,1) -- (1,0) -- (0.125,0) -- cycle;
\draw[fill=black!35] (1,1) -- (0,1) -- 
    (0,2) -- (2,2) -- (2,0) -- (1,0) -- cycle;
\draw[fill=black!30] (2,2) -- (0,2) -- 
    (0,3) -- (3,3) -- (3,0) -- (2,0) -- cycle;
\draw[fill=black!25] (3,3) -- (0,3) -- 
    (0,4) -- (4,4) -- (4,0) -- (3,0) -- cycle;
\draw[anchor=west] (0.1875,0) node[rotate=-90] {$2$};
\draw[anchor=west] (0.9375,0) node[rotate=-90] {$8$};
\draw[anchor=west] (1.9375,0) node[rotate=-90] {$16$};
\draw[anchor=west] (2.9375,0) node[rotate=-90] {$24$};
\draw[anchor=west] (3.9375,0) node[rotate=-90] {$32$};
\draw[overlay,anchor=east] (0,0.1875) node {$2$};
\draw[overlay,anchor=east] (0,0.9375) node {$8$};
\draw[overlay,anchor=east] (0,1.9375) node {$16$};
\draw[overlay,anchor=east] (0,2.9375) node {$24$};
\draw[overlay,anchor=east] (0,3.9375) node {$32$};
\draw (0.5,0.5) node[rotate=-45] {$10^8$};
\draw (0.5,1.5) node[rotate=-45] {$10^7$};
\draw (1.5,0.5) node[rotate=-45] {$10^7$};
\draw (0.5,2.5) node[rotate=-45] {$10^6$};
\draw (2.5,0.5) node[rotate=-45] {$10^6$};
\draw (0.5,3.5) node[rotate=-45] {$10^5$};
\draw (3.5,0.5) node[rotate=-45] {$10^5$};
\begin{scope}[shift={(4.7,0)}]
\draw[fill=black!40] (0.125,0.125) -- (0,0.125) -- 
    (0,1) -- (1,1) -- (1,0) -- (0.125,0) -- cycle;
\draw[fill=black!35] (1,1) -- (0,1) -- 
    (0,2) -- (2,2) -- (2,0) -- (1,0) -- cycle;
\draw[fill=black!30] (3,3) -- (3,0) -- (2,0) -- (2,1.5)
    \foreach \k in {1,...,4} { -- ++(-0.125,0) -- ++ (0,0.125) }
    -- (0,2) -- (0,3) -- cycle;
\draw[fill=black!25] (4,4) -- (4,0) -- (3,0) -- (3,1.5) -- (2,1.5)
    -- (2,2) -- (1.5,2) -- (1.5,3) -- (0,3) -- (0,4) -- cycle;
\draw[fill=black!20] (4,4) -- (4,1.5) -- (3,1.5) -- (3,1.75)
    \foreach \k in {1,...,10} { -- ++(-0.125,0) -- ++ (0,0.125) }
    -- (1.5,3) -- (1.5,4) -- cycle;
\draw[fill=black!15] (4,4) -- (4,1.75) -- (3,1.75)
    -- (3,3) -- (1.75,3) --(1.75,4) -- cycle;
\draw[anchor=west] (0.1875,0) node[rotate=-90] {$2$};
\draw[anchor=west] (0.9375,0) node[rotate=-90] {$8$};
\draw[anchor=west] (1.9375,0) node[rotate=-90] {$16$};
\draw[anchor=west] (2.9375,0) node[rotate=-90] {$24$};
\draw[anchor=west] (3.9375,0) node[rotate=-90] {$32$};
\draw[anchor=east] (1.4375,4) node[rotate=-90] {$12$};
\draw[anchor=east] (1.6875,4) node[rotate=-90] {$14$};
\draw[overlay,anchor=east] (0,0.1875) node {$2$};
\draw[overlay,anchor=east] (0,0.9375) node {$8$};
\draw[overlay,anchor=east] (0,1.9375) node {$16$};
\draw[overlay,anchor=east] (0,2.9375) node {$24$};
\draw[overlay,anchor=east] (0,3.9375) node {$32$};
\draw[overlay,anchor=west] (4,1.4375) node {$12$};
\draw[overlay,anchor=west] (4,1.6875) node {$14$};
\draw (0.5,0.5) node[rotate=-45] {$10^8$};
\draw (0.5,1.5) node[rotate=-45] {$10^7$};
\draw (1.5,0.5) node[rotate=-45] {$10^7$};
\draw (0.5,2.5) node[rotate=-45] {$10^6$};
\draw (2.5,0.5) node[rotate=-45] {$10^6$};
\draw (0.5,3.5) node[rotate=-45] {$10^5$};
\draw (3.5,0.5) node[rotate=-45] {$10^5$};
\draw (2.15,2.15) node[rotate=-45] {$10^5$};
\draw (2.6,2.6) node[rotate=-45] {$10^4$};
\draw (3.5,3.5) node[rotate=-45] {$10^3$};
\end{scope}
\end{tikzpicture}
\end{center}
For Shisen-Sho, the proportion of winnable boards is the same as for the transposed layout. This yielded double sample size for the non-square rectangular layouts. To get double sample size for square layouts as well, 
we simply chose the samples twice as large. So the effective sample sizes are $200\,000$, $2\,000\,000$, $20\,000\,000$ and $200\,000\,000$.

For Mahjongg Solitaire, we doubled the sample size by combining the
sample results with those of the transposed layouts for transposed Mahjongg Solitaire. So the effective sample sizes are $2\,000$, $20\,000$, $200\,000$, $2\,000\,000$, $20\,000\,000$ and $200\,000\,000$. A percentage in the graph indicates that there were both winnable and impossible boards.

\subsection{xmahjongg layouts}

Having a compute cluster this time, we sampled the xmahjongg layouts $200$ times as many as in \cite{dB}. So we sampled the default layout
$2\,000\,000\,000$ times and the other layouts $20\,000\,000$ times each. We sampled the default layout in $20$ threads of $100\,000\,000$. For Mahjong Solitaire, each chunk of $100\,000\,000$ boards took about $3$ weeks. For Shisen-Sho and transposed Mahjong Solitaire, this was less 
than $6$ days and $2$ days respectively.

The graphs for transposed Mahjong Solitaire are at the right of those for Mahjong Solitaire. For the symmetric deepwell layout, we combined the results for Mahjong Solitaire and transposed Mahjong Solitaire, resulting in a sample size of $40\,000\,000$ for Mahjong Solitaire. The foobar layout is just the foo layout, extended with the sample results of the bar layout with transposed matching rules, resulting in a sample size of $40\,000\,000$ for all three game types.

We obtained the following results for the default layout:
\begin{center}
\begin{tabular}{rl}
Mahjong Solitaire: & 2.959 percent impossible \\
Transposed Mahjong Solitaire: & 1.756 percent impossible \\
Shisen-Sho: & 1.906 percent impossible
\end{tabular}
\end{center}
The first $100\,000$ boards of the papillon layout are not winnable, so no winnable boards were found for this layout in \cite{dB}. But $16$ of the $20\,000\,000$ boards are winnable. The hourglass layout has no winnable boards in the sample of $20\,000\,000$ boards.

\end{document}